\documentclass{sig-alternate}

\usepackage[pdfauthor={Mark Giesbrecht, Albert Heinle, Viktor Levandovskyy},
          pdftitle={Factoring Differential Operators in n Variables}
             , bookmarksnumbered=true,colorlinks,linkcolor=darkblue,
             citecolor=darkgreen, urlcolor=darkred
            ]{hyperref}
\usepackage{url}

\usepackage[svgnames]{xcolor}
\definecolor{darkgreen}{rgb}{0,.35,0}
\definecolor{darkblue}{rgb}{0,0,.5}
\definecolor{darkred}{rgb}{.6,0,0}
\usepackage{enumitem}

\newcommand{\ZZ}{{\mathbb{Z}}} 
 
\newcommand{\NN}{{\mathbb{N}}} \newcommand{\KK}{{\mathbb{K}}}

\newcommand{\uX}{{\underline{X}}}
\newcommand{\uD}{{\underline{D}}}
\newcommand{\undlnn}{{\underline{n}}}
\newcommand{\uTheta}{{\underline{\theta}}}
\newcommand{\uZero}{{\underline{0}}}
\newcommand{\lm}{{\mathrm{lm}}}
\newcommand{\lexp}{{\mathrm{lexp}}}

\newtheorem{theorem}{Theorem}
\newtheorem{definition}{Definition}
\newtheorem{example}{Example}
\newtheorem{lemma}{Lemma}
\newtheorem{remark}{Remark}
\newtheorem{corollary}{Corollary}
\newtheorem{proposition}{Proposition}

\begin{document}

\setlength{\belowdisplayskip}{1pt} \setlength{\belowdisplayshortskip}{0pt}
\setlength{\abovedisplayskip}{1pt}
\setlength{\abovedisplayshortskip}{0pt}
\title{Factoring Differential Operators in {\huge{$n$}} Variables}

\numberofauthors{3} \author{
\alignauthor Mark Giesbrecht\\ \affaddr{David R. Cheriton School of
  Computer Science}\\ \affaddr{ University of Waterloo}\\ \affaddr{200
  University Avenue West}\\ \affaddr{Waterloo, Ontario,
  Canada}\\ \email{mwg@uwaterloo.ca}
\alignauthor Albert Heinle\\ \affaddr{David R. Cheriton School of
  Computer Science}\\ \affaddr{ University of Waterloo}\\ \affaddr{200
  University Avenue West}\\ \affaddr{Waterloo, Ontario,
  Canada}\\ \email{aheinle@uwaterloo.ca}
\alignauthor Viktor Levandovskyy\\ \affaddr{Lehrstuhl D f\"ur
  Mathematik, RWTH Aachen University}\\ \affaddr{Templergraben
  64}\\ \affaddr{Aachen,
  Germany}\\ \email{viktor.levandovskyy@math.rwth-aachen.de} }

\date{\today}

\maketitle
\begin{abstract}
In this paper, we present a new algorithm and an experimental
implementation for factoring elements in the polynomial $n$th Weyl
algebra, the polynomial $n$th shift algebra, and $\ZZ^n$-graded
polynomials in the $n$th $\underline{q}$-Weyl algebra.

The most unexpected result is that this noncommutative problem of
factoring partial differential operators can be approached effectively
by reducing it to the problem of solving systems of polynomial
equations over a commutative ring. In the case where a given
polynomial is $\ZZ^n$-graded, we can reduce the problem completely
to factoring an element in a commutative multivariate polynomial
ring.

The implementation in \textsc{Singular} is effective on a broad range of
polynomials and increases the ability of computer algebra systems to
address this important problem.  We compare the
performance and output of our algorithm with other implementations
in commodity computer algebra systems on  nontrivial examples.
\end{abstract}
\vspace*{-10pt}
\vspace*{-10pt}

\section{Introduction}
In this paper we present a new method and an implementation for
factoring elements in the $n$th polynomial Weyl algebra $A_n$ and the
$n$th polynomial shift algebra. An adaptions of these ideas can also
be applied to classes of polynomials in the $n$th $\underline{q}$-Weyl
algebra, which is also outlined here.  There are numerous important
applications for this method, notably since one can view those rings
as operator algebras. For example, given an element $L\in A_n$ and
viewing $L$ as a differential operator, one can derive
properties of its solution spaces. Especially concerning the problem
of finding the solution to the differential equation associated with
$L$, the preconditioning step of factoring $L$ can help to reduce the
complexity of that problem in advance.

The new technique heavily uses the nontrivial $\ZZ^n$-grading on $A_n$
and, to the best of our knowledge, has no analogues in the literature
on factorizations for $n\geq 2$.
However, for $n=1$ it is the same grading that lies behind the
Kashiwara-Malgrange $V$-filtration (\cite{Kashiwara:1983} and
\cite{Malgrange:1983}), which is a very important tool in the
$D$-module theory.  Van Hoeij also made use of this technique in
\cite{van1997formal} to factorize elements in the first Weyl algebra
with power series coefficients.  Notably, for $n \geq 2$, the
$\ZZ^n$-grading we propose is very different from the mentioned
$\ZZ$-grading. Among others, a recent result from
\cite{andres2013noncommutative} states that the Gel'fand-Kirillov
dimension \cite{gelfand1966corps} of the $0$\emph{th} graded part of
$\ZZ$-grading of $A_n$ is in fact $2n-1$. The Gel'fand-Kirillov
dimension of the whole ring $A_n$ is, for comparison, $2n$. The
$\uZero$\emph{th} graded part of the $\ZZ^n$-grading we propose has
Gel'fand-Kirillov dimension $n$.

\vspace*{-7pt}
\begin{definition}
  Let $A$ be an algebra over a field $\KK$ and $f\in A\setminus \KK$
  be a 
  polynomial. A {\bf nontrivial factorization} of $f$ is a tuple
  $(c,f_1,\ldots,f_m)$, where $c \in
  \KK\setminus\{0\}$, $f_1,\ldots,f_m \in A\setminus\{1\}$ are monic polynomials and
  $f= c \cdot f_1 \cdots f_m.$
\end{definition}
\vspace*{-7pt}
In general, we identify two problems in noncommutative factorization
for a given polynomial $f$: (i) finding one factorization of $f$, and
(ii) finding all possible factorizations  of~$f$.
Item (ii) is interesting since factorizations in noncommutative rings
are not unique in the classical sense (i.e., up to multiplication by a
unit), and regarding the problem of solving the associated
differential equation one factorization might be more useful than
another. We show how to approach both problems here.

A number of papers and implementations have been published in the
field of factorization in operator algebras over the past few
decades. Most of them concentrated on linear differential operators
with rational coefficients.  Tsarev has studied the form, number and
properties of the factors of a differential operator in
\cite{Tsarev:1994} and \cite{Tsarev:1996}, which extends the papers
\cite{Loewy:1903} and \cite{Loewy:1906}.  A very general approach to
noncommutative algebras and their properties, including factorization,
is also done in the book by Bueso et al. in \cite{Bueso:2003}. The
authors provide several algorithms and introduce various
points of views when dealing with noncommutative polynomial \mbox{algebras.}

In his dissertation van Hoeij \cite{Hoeij:1996} developed an algorithm
to factorize a univariate differential operator. Several papers following that
dissertation extend these techniques \cite{Hoeij:1997,
  van1997formal, Hoeij:2010}, and this algorithm is implemented in the
\texttt{DETools} package of \textsc{Maple} \cite{Maple} as the
standard algorithm for factorization of these operators.

In the \textsc{REDUCE}-base computer algebra system \textsc{ALLTYPES},
Schwarz and Grigoriev \cite{Schwarz:2009} have implemented the
algorithm for factoring differential operators they introduced in
\cite{GrigorievSchwartz:2004}.  As far as we know, this implementation
is solely accessible as a web service.  Beals and Kartashova
\cite{beals2005constructively} consider the problem of factoring
polynomials in the second Weyl algebra, where they are able to deduce
parametric \mbox{factors.}

For special classes of polynomials in operator algebras, Foupouagnigni
et al. \cite{foupouagnigni2004factorization} show some unexpected
results about factorizations of fourth-order differential equations
satisfied by certain Laguerre-Hahn orthogonal
polynomials.

From a more algebraic point of view, and dealing only with strictly
polynomial noncommutative algebras, Melenk and Apel \cite{Melenk:1994}
developed a package for the computer algebra system
\textsc{REDUCE}. That package provides tools to deal with
noncommutative polynomial algebras and also contains a factorization
algorithm for the supported algebras. It is, moreover, the only tool
besides our implementation in \textsc{Singular} \cite{Singular:2012}
that is capable of factoring in operator algebras with more than one
variable. Unfortunately, there are no further publications about how
the implementation works besides the available code.

The above mentioned algorithms and implementations are very well written and
they are able to factorize a large number of polynomials.
Nonetheless, as pointed out in \cite{HeinleLev:2011,  heinle2013factorization}, there exists a large class of polynomials,
even in the first Weyl algebra, that seem to form the worst case for
those algorithms.
This class is namely the graded (or homogeneous) polynomials in the
sense of the $\ZZ^n$-graded structure on the $n$th Weyl algebra. Using
our techniques, we are able to obtain a factorization very quickly
utilizing commutative factorization and some combinatorics. Those
techniques are discussed for the first ($\underline{q}$-)Weyl algebra
in detail in \cite{heinle2013factorization}.

Factorization of a general non-graded polynomial is much more
involved.  The main idea lies in inspecting the highest and the lowest
graded summands of the polynomial to factorize. Any factorization
corresponds respectively to the highest or the lowest graded
summands of the factors. Since the graded factorization appears to be
easy, we are able to factorize those summands and obtain finitely many
candidates for highest and lowest summands of the factors. Obtaining
the rest of the graded summands is the subject consider in this
paper.

An implementation dealing with the first Weyl algebra, the first shift
algebra, and graded polynomials in the first $q$-Weyl algebra, was
created by
Heinle and Levandovskyy within the computer algebra system
\textsc{Singular} \cite{Singular:2012}.  For the latter algebra, the
implementation in \textsc{Singular} is the only one available that
deals with $q$-Weyl algebras, to the knowledge of the
authors. The code has been distributed since version 3-1-3 of
\textsc{Singular} inside the library \texttt{ncfactor.lib}, and
received a major update in version 3-1-6.

The new approach described in this paper will soon be
available in an upcoming release of \textsc{Singular}.
There are new functions for factoring polynomials in the $n$th
polynomial Weyl algebra, homogeneous polynomials in the $n$th
polynomial $\underline{q}$-Weyl algebra and the $n$th polynomial shift
algebra.

The remainder of this paper is organized as follows. The rest of this
section is dedicated to providing basic notions, definitions and results
that are needed to describe our approach. Most of the results are
well-known, but have not been used for factorization until now.

Section \ref{sctn:factGraded} contains a methodology to deal with
the factorization problem for graded polynomials, while in
Section~\ref{sctn:factArb} we utilize this methodology to factorize
arbitrary polynomials in the $n$th Weyl algebra. 
In Section~\ref{sctn:conclusion} we evaluate our
experimental implementation on several examples in Section
\ref{sctn:implAndTime} and compare the results to other commodity
computer algebra systems.

\subsection{Basic Notions and Definitions}


By $\KK$ we always denote a field of characteristic zero (though some
of the statements also hold for some finite fields).  For notational
convenience we write $\undlnn$ for $\{1,\ldots,n\}$ and $\uTheta$
for $\theta_1,\ldots, \theta_n$ for $n \in \NN$ throughout.
\vspace*{-3pt}
\begin{definition}
  The \textbf{$n$th $\underline{q}$-Weyl algebra} $Q_{n}$ is defined~as
  \begin{align*}
    &Q_n := \mathbb{K}\biggl\langle  x_1, \ldots, x_n ,\partial_1, \ldots,
    \partial_n | \text{ for } (i,j) \in \undlnn\times
    \undlnn:\\ 
    &\partial_i x_j =
    \begin{cases}
      x_j \partial_i,  \text{ if } i \neq j\\ 
      q_i x_j \partial_i + 1,   \text{ if } i=j
    \end{cases}\hspace{-10pt} 
    \text{, }
    \partial_i \partial_j - \partial_j \partial_i = x_ix_j - x_j x_i
    = 0\biggr\rangle,
  \end{align*}
  where $q_1, \ldots, q_n$ are units in $\KK$.  For the special case
  where $q_1 = \dots = q_n = 1$ we have the \textbf{$n$th Weyl
    algebra}, which is denoted by $A_n$. For notational convenience,
  we write
  \(
  \uX^e \uD^w := x_1^{e_1}\cdots x_n^{e_n}\partial_1^{w_1} \cdots
  \partial_n^{w_n}
  \)
  for every monomial, where $e, w \in \NN_0^n$.
\end{definition}
\vspace*{-7pt}

\begin{definition}
  The \textbf{$n$th shift algebra} $\mathcal{S}_n$ is defined as
  \begin{align*}
    &\mathcal{S}_n := \mathbb{K}\biggl\langle  x_1, \ldots, x_n ,s_1,
    \ldots, s_n | \text{ for } (i,j) \in \undlnn\times
    \undlnn:\\
    &s_i x_j =
    \begin{cases}
      x_j s_i,  \text{ if } i \neq j\\ (x_j+1) s_i, \text{ if } i=j
    \end{cases}\hspace{-10pt} \text{, }  s_i s_j - s_j s_i = x_ix_j - x_j x_i = 0\biggr\rangle.
  \end{align*}
  For notational convenience, we write as above \linebreak
  \(
  \uX^e \underline{S}^w := x_1^{e_1}\cdots x_n^{e_n}s_1^{w_1} \cdots
  s_n^{w_n},
  \)
  where $e,w \in \NN_0^n$.
\end{definition}


\begin{remark}
\label{rem:OrderingOnZn}
Throughout this paper we view $\ZZ^n$, equipped with the
coordinate-wise addition, as an ordered 
monoid with respect to a total ordering $<$, compatible
with addition and satisfying the following
property: for any $z_1, z_2 \in \ZZ^n$, such that $z_2 < z_1$, 
the set $\{ w \in \ZZ^n, z_2 < w < z_1 \}$ is finite.
\end{remark}

The $n$th ($\underline{q}$-)Weyl algebra possesses a nontrivial
$\ZZ^n$-grading 
using the weight vector $[-w,w]$ for a $\underline{0} \neq w \in \ZZ^n$ on the elements
$x_1,\ldots, x_n, \partial_1, \ldots, \partial_n$. For simplicity, we
choose $w := [1, \ldots, 1]$. In what follows, $\deg$ denotes the
degree induced by this weight vector, that is $\deg(\uX^a \uD^b) := 
[b_1 - a_1, \ldots, b_n - a_n]$ for $a,b\in\NN_{0}^n$. Note, that a $\ZZ$-grading, arising from
$V$-filtration \cite{Kashiwara:1983, Malgrange:1983} prescribes 
to $\uX^a \uD^b$ the grade $\sum_{i=1}^n (b_i - a_i) \in \ZZ$.

We call a polynomial {\bf homogeneous} or {\bf graded}, if every summand is
weighted homogeneous with respect to the weight vector $[-w,w]$ as above.
\vspace*{-5pt}
\begin{example}
  In the second Weyl algebra $A_2$ one has
  \[
  \deg(x_1x_2\partial_1\partial_2) = \deg((\partial_1 x_1 +
  1)x_2\partial_2) = [0,0].
  \]
  The polynomial \( x_1\partial_1^2x_2 + x_1^4\partial_1^5x_2
  + \partial_1x_2 \) is homogeneous of degree $[1,-1]$. The monomials
  $x_1 \partial_2$, resp. $x_2 \partial_1$, have degrees $[-1,1]$,
  resp. $[1,-1]$, hence their sum is not homogeneous.
\end{example}


Note that the so-called \textbf{Euler operators} $\theta_i :=
x_i\partial_i, i\in \undlnn$ have degree $\underline{0}$ for all $i$,
and thus play an important role, as we shall soon see.

First, we study some commutation rules the Euler operator
$\theta_i$ has with $x_i$ and $\partial_i$. For $Q_n$, in order
abbreviate the size of our formulae, we introduce the so called
$q$-bracket.
\vspace*{-5pt}
\begin{definition}
  For $n \in \NN$ and $q \in \KK\setminus\{0\}$,  
  the {\bf $q$-bracket of $n$} 
  is defined to be
  \(
  [n]_q := \frac{1-q^n}{1-q} = \sum_{i = 0}^{n-1}q^i.
  \)
\end{definition}
\vspace*{-5pt}
\begin{lemma}[Compare with \cite{SaStuTaka:2000}]
  \label{lem:rewriteKTheta}
  In $A_n$, the following commutation rules hold for $m \in \NN$ and
  $i \in \undlnn$:
  \begin{align*}
    \theta_i x_i^m = x_i^m (\theta_i + m), \quad  \theta_i \partial_i^m  =\partial_i^m(\theta_i - m).
  \end{align*}
  More generally, in $Q_n$, the following commutation rules hold for
  $m \in \NN$ and $i \in \undlnn$:
  \[
    \theta_i x_i^m = x_i^m (q_i^m\theta_i + [m]_{q_i}),~
    \theta_i  \partial_i^m  =  \frac{\partial_i^m}{q_i}\left(
    \frac{\theta_i-1}{q_i^{m-1}} - \frac{q_i^{-m+2}-q_i}{1-q_i}
    \right).
  \]
\end{lemma}


The commutation rules described in Lemma \ref{lem:rewriteKTheta} can,
of course, be extended to arbitrary polynomials in the $\theta_i$,
$i\in \undlnn$.

\vspace*{-6pt}
\begin{corollary}
  \label{cor:thetaswap}
  Consider $f(\uTheta)\in
  \KK[\uTheta]$.
Then, in $A_n$ we have
\( f(\uTheta)\uX^e = \uX^e f(\theta_1 + e_1,
  \ldots, \theta_n + e_n),\) and 
  \(f(\uTheta)\uD^e =
  \uD^e f\left( \theta_1 -e_1, \ldots, \theta_n-e_n\right).\)
Analogous identities with the respective commutation rules as given in
Lemma \ref{lem:rewriteKTheta} hold for~$Q_n$.
\end{corollary}


\vspace*{-12pt}
\section{Factoring Graded Polynomials}
\label{sctn:factGraded}

For graded polynomials, the main idea of our factorization technique
lies in the reduction to a commutative univariate polynomial subring
of $A_n$, respectively $Q_n$, namely $\KK[\uTheta]$.  Actually, it
appears that this subring is quite large in the sense of reducibility
of its elements in $A_n$ (resp. $Q_n$).

Due to the commutativity of $x_i$ with $\partial_j$, for $i\neq j$, we
can write
\(
\uX^a \uD^b = 
x_1^{a_1} \cdots x_n^{a_n} \cdot \partial_1^{b_1} \cdots \partial_n^{b_n}=
x_1^{a_1} \cdot \partial_1^{b_1} \cdots  x_n^{a_n} \cdot \partial_n^{b_n}
\) for any $a,b \in \NN_0^n$.
By definition, a monomial $\uX^a \uD^b$ has degree $\uZero:=[0,
\ldots, 0]$ if and only if $a=b$. The following lemma shows, how we
can rewrite every homogeneous polynomial of degree $\uZero$ in $A_n$
(resp. $Q_n$) as a polynomial in $\KK[\uTheta]$.

\vspace*{-5pt}
\begin{lemma}[Compare with \cite{SaStuTaka:2000}, Lemma 1.3.1]
  \label{lem:homogToTheta}
  In $A_n$, we have the identity \(
  x_i^m\partial_i^m = \prod_{j = 0}^{m-1} (\theta_i - j).
  \) for $m \in \NN$ and $i \in \undlnn$.
  In $Q_n$, one can rewrite $x_i^m\partial_i^m$ as element in $\KK[\uTheta]$
  and it is equal to
  \(\frac{1}{q_i^{T_{m-1}}}\prod_{j=0}^{m-1}\left(\theta_i
    - [j]_{q_i}\right),
  \)
  where $T_j$ denotes the $j$th triangular number, i.e.,  $T_j :=j(j+1)/2$ for all $j \in \NN_0$.
\end{lemma}

\vspace*{-8pt}
\begin{corollary}
The $\uZero$th graded part of $Q_n$, respectively $A_n$, is
$\KK[\theta_1, \ldots, \theta_n]$.
\end{corollary}
\vspace*{-5pt}

Recall, that the $z$th graded part for $z \in \ZZ^n$ of $Q_n$, resp. $A_n$,
is defined to be the $\KK$-vector space:
$$Q_n^{(z)} := \KK \{\uX^{n_1}\uD^{n_2} : n_1,n_2\in \NN_0^n, \ n_2-n_1=z \},$$
i.e., the degree of a monomial is determined by the difference of its
powers in the $x_i$ and the $\partial_i$. Moreover, since in a grading 
$Q_n^{(z_1)} \cdot Q_n^{(z_2)} \subseteq Q_n^{(z_1+z_2)}$ holds
for all $z_1, z_2 \in \ZZ^n$, $Q_n^{(z)}$ is naturally a $Q_n^{(\uZero)}$-module.

\vspace*{-6pt}
\begin{proposition}
  \label{prop:A0completeWeylDescription}
  For $z \in \ZZ^n\setminus\{\uZero\}$, the $z$th graded part
  $Q_n^{(z)}$, resp. $A_n^{(z)}$, is a cyclic $\KK[\uTheta]$-bimodule,
  generated by the element $\uX^{e(z)} \uD^{w(z)}$, exponent vectors
  of which read for $i \in\undlnn$ as follows:
  \begin{align*}
    e_i(z) :=
    \begin{cases}
      -z_i, & \text{if } z_i <0,\\
      0, & \text{otherwise},
    \end{cases},\quad
    w_i(z) :=
    \begin{cases}
      z_i, & \text{if } z_i >0,\\
      0, & \text{otherwise}.
    \end{cases}
  \end{align*}
\end{proposition}
\begin{proof}
  
  A polynomial $p \in  Q_n^{(z)}$ resp. $p \in A_n^{(z)}$ is homogeneous of degree $z \in \ZZ^n$ if and only if every monomial of $p$ is of the form
  $\uX^{\overline{k} + e(z)}\uD^{\overline{k} + w(z)}$, where $k \in
  \NN_0$ and $\overline{k} := [k,\ldots, k]$.  By doing a rewriting, similar to
  the above, we obtain
  \(
  \uX^{\overline{k} + e(z)}\uD^{\overline{k} + w(z)} = 
   \uX^{e(z)}\uX^{\overline{k}}\uD^{\overline{k}}\uD^{w(z)} = 
    \uX^{e(z)}f_{k}(\underline{\theta})\uD^{w(z)},
  \)
  where $f_{k}(\underline{\theta})$ is computed utilizing Lemma \ref{lem:homogToTheta}. Moreover, by Corollary \ref{cor:thetaswap},
  we conclude that 
  \[
  \uX^{e(z)} f_{k}(\underline{\theta})\uD^{w(z)} = 
  f_{k}(\theta_1 -e_1(z),\ldots,\theta_n -e_n(z))\uX^{e(z)}\uD^{w(z)}
  \]
  or, equivalently, $\uX^{e(z)}\uD^{w(z)} f_{k}(\theta_1 +w_1(z),\ldots,\theta_n +w_n(z))$, showing the cyclic bimodule property.
\end{proof}

Therefore, the factorization of a homogeneous polynomial of degree
zero can be done by rewriting the polynomial as an element in
$\KK[\uTheta]$ and applying a commutative factorization on the
polynomial, a much-better-understood problem which is also well
implemented in every computer algebra system.

Of course, this would not be a complete factorization, as there are
still elements irreducible in $\KK[\uTheta]$ which are reducible in $Q_n$,
resp. $A_n$. An obvious example are the $\theta_i$ themselves.
Fortunately, there are only $2n$ monic polynomials irreducible
in $\KK[\uTheta]$ that are reducible in $A_n$, resp. $Q_n$, and these
are of quite a special form.  This
extends the proof for $A_1$ and $Q_1$ presented in
\cite{heinle2013factorization}.
\vspace*{-4pt}
\begin{lemma}
\label{lem:thetairred}
Let $i \in \undlnn$.
The polynomials $\theta_i$ and $\theta_i + \frac{1}{q_i}$ are the only
irreducible monic elements in $\KK[\uTheta]$ that are reducible in
$Q_n$. Respectively, $\theta_i $ and $\theta_i +1$ are the only irreducible monic polynomials in $\KK[\uTheta]$ that are reducible in $A_n$.
\end{lemma}
\begin{proof}
 We only consider the proof for $A_n$, as the proof for $Q_n$ is
 done in an analogous way.
 Let $f \in \KK[\uTheta]$ be a monic polynomial. Assume that it is
 irreducible in $\KK[\uTheta]$, but reducible in $A_n$.
 Let $\varphi,\psi$ be elements in $A_n$ with $\varphi\psi = f$. Then
$\varphi$ and $\psi$ are homogeneous and $\varphi \in A_n^{(-z)}, \psi
 \in A_n^{(z)}$ for a $z\in \ZZ^n$. 
 Let $[e,w]:=[e(z), w(z)]$ be as in Proposition \ref{prop:A0completeWeylDescription}. Note, that then $w(-z) = e(z) = e$ and $e(-z)=w(z)=w$ holds. That is, $A_n^{(z)} = \KK[\uTheta]\uX^{e}\uD^{w}$ whereas $A_n^{(-z)} = \KK[\uTheta]\uX^{w}\uD^{e}$.
Then for 
$\tilde \varphi, \tilde \psi \in \KK[\uTheta] $, we have
$ \varphi = \tilde \varphi(\uTheta) \uX^{e}\uD^{w}$ and $\psi =  \tilde \psi (\uTheta)\uX^{w}\uD^{e}.$
 Using Corollary \ref{cor:thetaswap}, we obtain
\begin{align*}
f = \tilde \varphi (\uTheta) \uX^{e}\uD^{w} \tilde \psi
(\uTheta) \uX^{w}\uD^{e}
   = \tilde
 \varphi (\uTheta) \uX^{e}\uD^{w}\uX^{w}\uD^{e}
 \tilde \psi (\uTheta + w - e),
 \end{align*}
 where, by Lemma \ref{lem:homogToTheta}, $\uX^{e}\uD^{w}\uX^{w}\uD^{e}
 = g(\uTheta) \in \KK[\uTheta]$.  Since vectors $e$ and $w$ have
 disjoint support and $e+w=[|z_1|,\ldots,|z_n|]$, $g$ is
 irreducible by Lemma \ref{lem:homogToTheta} only if there is at
 most one nonzero $z_i$. If $z=\uZero$, then $e=w=0$, hence $g = 1$
 and $\phi, \psi \in \KK[\uTheta]$. Since $f$ has been assumed to be
 monic irreducible in $\KK[\uTheta]$, one $\phi$ and $\psi$ give us a
 trivial factorization.

  Now, suppose that there exists exactly one $i$ such that $z_i > 0$. 
Then $e(z)=0$ and $w(z)=z$ is zero on all but $i$th place. By the
irreducibility assumption on $f \in \KK[\uTheta]$ we must have
$\tilde \varphi,\tilde \psi \in \KK$; since $f$ is monic, we must
also have $\tilde\varphi = \tilde\psi^{-1}$.  By Lemma
\ref{lem:homogToTheta} we obtain $z_i = 1$.  As a result, the only
possible $f$ in this case is $f=\theta_i + 1$. For analogous reasons
for the case when $z_i <0$, we conclude, that the only possible $f$
in that case is $f = \theta_i$. 
\end{proof}
\vspace*{-5pt}

The result in Lemma \ref{lem:thetairred} provides us with an easy way to factor a homogeneous polynomial $p
\in A_n$, resp. $p \in Q_n$, of
degree $\uZero$. Obtaining one possible
factorization into irreducible polynomials can be done using the following steps:

\begin{enumerate}[topsep=2pt,itemsep=0pt,parsep=2pt]
  \item Rewrite $p$ as an element in $\KK[\uTheta]$;
  \item Factorize this resulting element in $\KK[\uTheta]$ with
    commutative methods;
  \item If there is $\theta_i$ or $\theta_i+1$ for $i \in \undlnn$ among the factors,
    rewrite it as $x_i \cdot \partial_i$ resp. $\partial_i\cdot x_i$.
\end{enumerate}

As mentioned earlier, the factorization of a polynomial in a
noncommutative ring is unique up to a weak
similarity~\cite{Bueso:2003}.  This notion is much more involved than
the similarity up to multiplication by units or up to interchanging
factors, as in the commutative case.  Indeed, several different
nontrivial factorizations can occur.  Fortunately, in the case of the
polynomial first ($\underline{q}$-)Weyl algebra, there are only
finitely many different nontrivial factorizations possible due to
\cite{Tsarev:1996}. In order to obtain all these different
factorizations, one can apply the commutation rules for $x_i$ and
$\partial_i$ with $\theta_i$ for $i \in \undlnn$. That these are all
possible factorizations up to multiplication by units can be seen
using an analogue approach as in the proof of Lemma
\ref{lem:thetairred}. Consider the following example.

\vspace*{-7pt}
\begin{example}
  Let
  \(
  p :=
x_1^2x_2\partial_1^2\partial_2+2x_1x_2\partial_1\partial_2+x_1\partial_1+1
\in A_2.
  \)
  The polynomial $p$ is homogeneous of degree $\uZero$, and hence
  belongs to $\KK[\uTheta]$ as
  \(
  \theta_1(\theta_1-1)\theta_2 + 2\theta_1\theta_2 + \theta_1 + 1.
  \)
This polynomial factorizes in $\KK[\uTheta]$ into $(\theta_1\theta_2 +
1)(\theta_1+1)$.
 Since $\theta_1+1$ factorizes as $\partial_1 \cdot x_1$, 
we obtain the following possible different nontrivial factorizations:
\[
(\theta_1\theta_2 + 1)\cdot \partial_1 \cdot x_1 = 
\partial_1 \cdot ((\theta_1-1)\theta_2 + 1) \cdot x_1=
\partial_1\cdot x_1 \cdot (\theta_1\theta_2 + 1).
\]
Note that $x_1 \partial_1 +1$ is not irreducible, since it factorizes nontrivially as  $\partial_1 \cdot x_1$.
\end{example}

\vspace*{-3pt}
Now we consider the factorization of homogeneous polynomials of
arbitrary degree $z \in \ZZ^n$. Fortunately, the hard work is already done
in Proposition \ref{prop:A0completeWeylDescription}.
Indeed, one factorization of a homogeneous polynomial $p \in
Q_n^{(z)}$, resp. $p \in A_n^{(z)}$, of degree $z \in \ZZ^n$ can be
obtained using the following steps.
\begin{enumerate}[itemsep=0pt,topsep=2pt,itemsep=0pt]
\item Write $p(\uX,\uD)$ as ${\tilde p}(\uTheta) \uX^e \uD^w$,
  where $\tilde p \in A_n^{(\uZero)}=\KK[\uTheta]$ and $e,w$ are
  constructed according to
  Proposition \ref{prop:A0completeWeylDescription}.
\item Factorize $\tilde p$ using the technique described for
  polynomials of degree $\uZero$. Append to such a factorization the
  natural expansion of the monomial $\uX^e \uD^w$ into the product of occuring
  single variables.
\end{enumerate}

This leads to one nontrivial factorization. A characterization of
how to obtain all factorizations is given provided by the following lemma.

\vspace*{-6pt}
\begin{lemma}
Let $z \in \ZZ^n$ and $p \in A_n^{(z)}$, resp. $p \in Q_n^{(z)}$, is
  monic.  Suppose, that one factorization has been constructed as
  above and has the form $Q(\uTheta) \cdot T(\uTheta) \cdot \uX^e
  \uD^w$, where 
  $T(\uTheta) = \prod_{i=1}^{n} (x_i \partial_i)^{t_i} (\partial_j x_j)^{s_i}$
  is a product of irreducible factors in $\KK[\uTheta]$, which are
  reducible in $A_n$, resp. $Q_n$,
  and $Q(\uTheta)$ is the product of irreducible factors
  in both $\KK[\uTheta]$ and $A_n$, resp. $Q_n)$.
  Let $p_1 \cdots p_m$ for $m\in \NN$ be another nontrivial
  factorization of $p$. Then this factorization can be derived from
  $Q(\uTheta) \cdot T(\uTheta) \cdot \uX^e \uD^w$ by using two
  operations, namely (i) ``swapping'', that is interchanging two
  adjacent factors according to the commutation rules and (ii)
  ``rewriting'' of occurring $\theta_i$ resp. $\theta_i +1$ 
    ($\theta_i + \frac{1}{q}$ in the $\underline{q}$-Weyl case) by
    $x_i \cdot \partial_i$ resp. $\partial_i \cdot  x_i$.
\end{lemma}
\begin{proof}
Since $p$ is homogeneous, all $p_i$ for $i \in \underline{m}$ are
 homogeneous, thus 
 each of them can be written in the form
 $p_i = \tilde p_i(\uTheta)\cdot \uX^{e^{(i)}} \uD^{w^{(i)}}$, where
 $e^{(i)}, w^{(i)} \in \NN_0^{n}$. With respect to the commutation
 rules as stated in Corollary \ref{cor:thetaswap}, we can swap the
 $\tilde p_i(\uTheta)$ to the left for any $2 \leq i \leq m$. 
 Note that it is possible for
 them to be transformed to the form $\theta_j$ resp. $\theta_j + 1$
 ($\theta_j + \frac{1}{q}$ in the $\underline{q}$-Weyl case), $j \in
 \undlnn$, after performing these swapping steps.  I.e., we have
 commuting factors, both belonging to $Q(\uTheta)$, as well as to
 $T(\uTheta)$ at the left. Our resulting product is thus $\tilde
 Q(\uTheta) \tilde T(\uTheta) \prod_{j=1}^m \uX^{e^{(j)}}
 \uD^{w^{(j)}}$, where the factors in $\tilde Q(\uTheta)$,
 resp. $\tilde T(\uTheta)$, contain a subset of the factors of
 $Q(\uTheta)$ resp. $T(\uTheta)$. By our assumption of $p$ having
 degree $z$, we are able to swap $\uX^e \uD^w$ to the right in $F
 :=\prod_{j=1}^m \uX^{e^{(j)}} \uD^{w^{(j)}}$, i.e., $F = \tilde F
 \uX^e \uD^w$ for $\tilde F \in A_n^{(0)}$. This step may involve
 combining some $x_j$ and $\partial_j$ to $\theta_j$ resp. $\theta_j
 + 1$, $j \in \undlnn$. Afterwards, this is also done to the
 remaining factors in $\tilde F$ that are not yet polynomials in
 $\KK[\uTheta]$ using the swapping operation. These polynomials are
 the remaining factors that belong to $Q(\uTheta)$, resp. $T(\uTheta)$,
 and can be swapped commutatively to their respective positions. Since
 reverse engineering of those steps is possible, we can derive the
 factorization $p_1 \cdots p_m$ from $Q(\uTheta) \cdot T(\uTheta)
 \cdot \uX^e \uD^w$ as claimed.
\end{proof}

Summarizing, we are now able to effectively factor graded polynomials in the
$n$th Weyl and $\underline{q}$-Weyl algebra. All different factorizations are obtainable using our
technique.

\begin{remark}
  A reader might ask what are the merits of our ``graded-driven''
  approach as opposed to a somewhat more direct approach to
  factorization using leading monomials.  Since, for monomials $m,
  m'\in A_n$, one has $\lexp(m \cdot m') = \lexp(m) + \lexp(m')$,
  indeed $h = p\cdot q$ implies $\lexp(p) + \lexp(q) = \lexp(h)$. Thus
  by considering, say, degree lexicographic ordering on $A_n$, the set
  $C_h := \{(a,b) \in \NN^n \times \NN^n : a, b \neq \uZero, \ a + b =
  \lexp(h) \}$ contains all possible pairs of leading monomials of $p$
  and $q$. Then, since with respect to the chosen ordering, for any
  monomials there are only finitely many smaller monomials, one can
  make an ansatz with unknown coefficients for $p$ and $q$. Each
  $(a,b) \in C_h$ leads to a system of nonlinear polynomial equations
  in finitely many variables.

  We compare this ``leading monomial'' approach with our ``graded-driven''
  one. At first, the factorization of a $\ZZ^n$-graded polynomial,
  which is very easy to accomplish with our approach, requires solving
  of several systems within the leading monomial approach. Second, the
  number of all elements in the set $C_h$ above is significantly
  bigger than the number of factorizations of the highest graded part
  of a polynomial, say ${\tilde p}(\uTheta)\uX^e \uD^w$: suppose that
  ${\tilde p}(\uTheta)$ is irreducible over $\KK[\uTheta]$. Then
  factorization with the ``graded-driven'' approach are obtained via
  moving $x$, resp. $\partial$, past ${\tilde p}(\uTheta)$ to the
  left. Thus the number of such factorizations is much smaller than
  the number of ways of writing the exponent vector of $\lm({\tilde
    p}(\uTheta)\uX^e \uD^w) = \uTheta^{\alpha} \uX^e \uD^w$ as a sum
  of two exponent vectors.
\end{remark}

In the next section, we show how the developed technique helps us to tackle
the factorization problem for arbitrary polynomials in $A_n$.

\section{Factoring Arbitrary \\ Polynomials}
\label{sctn:factArb}
\subsection{Preliminaries}

The techniques described in this section solve the factorization
problem in $A_n$. A generalization for $Q_n$ is more involved and the
subject of ongoing research.

We begin by fixing some notation used throughout
this section. From now on, let ``$<$'' be an ordering on $\ZZ^n$
satisfying the conditions of Remark \ref{rem:OrderingOnZn}.
Let $h \in A_n$ be the polynomial we want to factorize. As we are
deducing information from the graded summands of $h$, let furthermore $M
:= \{z^{(1)}, \ldots, z^{(m)}\}$, where $m\in \NN$ and $z^{(1)}> \ldots >
z^{(m)}$, be a finite subset of $\ZZ^n$ containing the degrees of
those graded summands. Hence, $h$ can be written in the
form $h =
\sum_{z \in M}h_z \in A_n$, where $h_z \in A_n^{(z)}$ for $z \in M$. Let us assume that $h$ possesses a nontrivial factorization
of at least two factors, which are not necessary
irreducible. Moreover, we assume that $m>1$, which means that $h$ is
not graded, since we have dealt with graded polynomials in $A_n$
already.  Let us denote the factors by
\begin{align}
  \label{eq:hFactor} 
  h = \sum_{z \in M}h_z:= \underbrace{(p_{\eta_1} + \ldots 
    + p_{\eta_k})}_{:=p}\underbrace{(q_{\mu_1}+ \ldots+
    q_{\mu_l})}_{:=q}, 
\end{align}
where $\eta_1>\eta_2>\ldots>\eta_k$ and $\mu_1>\mu_2>\ldots>\mu_l \in
\ZZ^n$, $p_{\eta_i} \in A_n^{(\eta_i)}$ for all $i\in \underline k$,
$q_{\mu_j}\in A_n^{(\mu_j)}$ for all $j \in \underline l$. We assume
that $p$ and $q$ are not graded, since we could easily
obtain those factors by simply comparing all factorizations of the
graded summands in $h$. In general, while trying to find a
factorization of $h$, we assume that the values of $k$ and $l$ are not
known to us beforehand. We will soon see how we can obtain
them. One can easily see that
\(
h_{z^{(1)}}=p_{\eta_1}q_{\mu_1}\) and \( h_{z^{(m)}} =
p_{\eta_k}q_{\mu_l},
\)
as the degree-wise biggest summand of $h$ can only be combined by 
multiplication of the highest summands of $p$ and $q$; analogously this
holds for the degree-wise lowest summand.

A finite set of candidates for $p_{\eta_1}, q_{\mu_1}, p_{\eta_k}$ and $q_{\mu_l}$ can be
obtained by factoring $h_{z^{(1)}}$ and $h_{z^{(m)}}$ using the
technique described in the previous section. Since the set of
candidates is finite, we can assume that the correct representatives
for $p_{\eta_1}, q_{\mu_1}, p_{\eta_k}$ and $q_{\mu_l}$ are known to
us. In practice, we would apply our method to all candidates and would
succeed in at least one case to factorize the polynomial due to our
assumption of $h$ being reducible.

One may ask now how many valid degrees could occur in summands of such
factors $p$ and $q$, i.e., what are the values of $l$ and
$k$. Theoretically, there exist choices for $\eta_1$ and $\eta_k$
(resp. $\mu_1$ and $\mu_l$) where there are infinitely many $z \in
\ZZ^n$ such that $\eta_1 > z > \eta_k$ (resp. $\mu_1 > z >
\mu_l$). Fortunately, only finitely many are valid degrees that can
appear in a factorization, as the next lemma shows.

\vspace*{-5pt}
\begin{lemma}
  \label{lem:finiteDegs}
  For fixed $h, p_{\eta_1}, q_{\mu_1}, p_{\eta_k}$ and $q_{\mu_l}\in
  A_n$ fulfilling the assumptions stated above, there are
  only finitely many possible $\eta_i$ resp. $\mu_j \in \ZZ^n, i,j \in
  \NN$, that can appear
  as degrees for graded summands in $p$ and $q$.
\end{lemma}
\begin{proof}
  For every variable $v\in \{x_1, \ldots, x_n, \partial_1,
  \ldots, \partial_n\} \subset A_n$, there exists a $j \in \NN_0$
  that represents the maximal degree of $v$ that occurs among the
  monomials in $h$. The number $j$ can be seen as a lower bound of the
  associated position of $v$ in $\eta_i$, resp. $\mu_i$, if $v$ is one
  of the $x$s, or as an upper bound if $v$ is one of the
  $\partial$s. If the degree of one of the homogeneous summands of $p$
  or $q$ would go higher, resp. lower, than this degree-bound, $v$ would
  appear in $h$ in a higher degree than $j$, which contradicts our
  choice of $j$. 
\end{proof}

\vspace*{-8pt}
\begin{example}
  \label{ex:finiteDegs}
 Let us consider
$$h = \underbrace{x_2\partial_1 \partial _2 + \partial_1}_{\text{degree: }[1,0]} + \underbrace{x_1
 x_2 \partial_1^2}_{\text{degree: } [1,-1]} +
\underbrace{4\partial_2}_{\text{degree: } [0,1]} +
\underbrace{4x_1\partial_1}_{\text{degree: } [0,0]} \in A_2.$$
One possible factorization of $x_2\partial_1 \partial _2 + \partial_1$ is
$\partial_2\cdot x_2\partial_1 =: p_{\eta_1} \cdot q_{\mu_1}$ and, on the other end,
one possible factorization of $4x_1\partial_1$ is $x_1 \partial_1\cdot 4
=: p_{\eta_k}\cdot q_{\mu_l}$. Concerning $p$, there are no elements
in $\ZZ^n$ that can occur between $\deg(p_{\eta_1}) = [0,1]$ and
$\deg(p_{\eta_k}) = [0,0]$; therefore we can set $k=2$. For $q$,
the only degree that can occur
between $\deg(q_{\mu_1}) = [1,-1]$  and $\deg(q_{\mu_l}) = [0,0]$
is $[0,1]$, as every variable except $\partial_1$ appears with maximal degree
$1$ in $h$. We have $l = 3$ in this case.
\end{example}



\subsection{Reduction to a Commutative System}

In the previous subsection we saw that, given $h \in A_n$ that
possesses a factorization as in (\ref{eq:hFactor}), we are able to
obtain the elements $p_{\eta_1}, q_{\mu_1}, p_{\eta_k}$ and
$q_{\mu_l}$. Moreover, we can compute the numbers $k$ and $l$ of homogeneous
summands in the factors. Now our goal is to find
values for the unknown homogeneous summands, i.e. $p_{\eta_2}, \ldots,
p_{\eta_{k-1}}, q_{\mu_2}, \ldots, q_{l-1}$. Our goal is to reduce
this to a commutative problem to the greatest extent we can.

For this, we use Proposition \ref{prop:A0completeWeylDescription}
and define for all $i \in \underline k$ the polynomial $\tilde
p_{\eta_i} \in A_n^{(\uZero)}$ by $\tilde p_{\eta_i}\uX^e\uD^w = p_{\eta_i}$. In the same way we define $\tilde q_{\mu_j}$ for
all $j \in \underline l$ and $\tilde h_z$ for $z \in M$.
The latter are known to us since $\tilde h_z$ can easily
be obtained from the input polynomial $h$. We can refer to the
$\tilde{h}_z$, $\tilde{p}_{\eta_i}$, $\tilde{q}_{\mu_j}$ as elements
in the commutative ring $\KK[\uTheta]$ using Lemma~\ref{lem:homogToTheta}.

The next fact about the degree of the remaining unknowns can be easily
proven and is
useful for our further steps.
\vspace*{-5pt}
\begin{lemma}
  \label{lem:degBoundTheta}
  The degree of the $\tilde p_{\eta_i}$ and the $\tilde q_{\mu_j}$,
  $(i,j) \in \underline{k} \times \underline{l}$, in
  $\theta_{t}$, $t \in \undlnn$, is bounded by
  $\min\{\deg_{x_t}(h),\deg_{\partial_t}(h) \}$, where $\deg_{v}(f)$
  denotes the degree of $f \in A_n$ in the variable $v$.
\end{lemma}
\vspace*{-5pt}

There are certain equations that the $\tilde p_{\eta_i}$ and the
$\tilde q_{\mu_j}$ must fulfil in order for $p$ and $q$ to be
factors of $h$.
\vspace*{-5pt}
\begin{definition}
  For $\alpha,\beta \in \ZZ^n$ we define
  \(
  \gamma_{\alpha,\beta} = \prod_{\kappa = 1}^n \tilde
  \gamma_{\alpha_\kappa,\beta_\kappa}^{(\kappa)}
  \);
  in the latter expression we define for $a, b \in \ZZ$ and $\kappa
  \in \undlnn$
  {{
  \[
  \tilde \gamma_{a,b}^{(\kappa)} := 
  \begin{cases}
    1, &\text{if } a,b\geq0 \lor a,b \leq
    0,\\ \prod_{\tau = 0}^{|a|-1}(\theta_{\kappa} - \tau), & \text{if } a<0,
    b>0, |a|\leq|b|,\\ \prod_{\tau = 0}^{|b|-1}(\theta_{\kappa} - \tau
    -|a|+|b|), & \text{if } a<0, b>0, |a|>|b|,\\ \prod_{\tau =
      1}^{a}(\theta_{\kappa} + \tau),& \text{if } a>0, b<0, |a| \leq |b|,\\
    \prod_{\tau = 1}^{|b|}(\theta_{\kappa} + \tau+|a|-|b|),& \text{if } a>0,
    b<0, |a| > |b|.
  \end{cases}
  \]
  }}
\end{definition}

\begin{theorem}
  \label{MainThm}
  Suppose that, with the notation as above, we have $h = pq$ and
  $\tilde p_{\eta_1}, \tilde q_{\mu_1}, \tilde p_{\eta_k}, \tilde h_{z^{(1)}}, \ldots,
  \tilde h_{z^{(m)}}$ are known. Define $\tilde h_{z} := 0$ for $z^{(1)}
  > z > z^{(m)}$ and $z \not\in M$. Then the remaining
  unknown $\tilde p_{\eta_2}, \ldots , \tilde
  {p}_{\eta_{k-1}},$ $ \tilde q_{\mu_2}, \ldots , \tilde
  {q}_{\mu_{l-1}}$ are solutions of the following finite set of
  equations:
  \begin{align}
    \label{eq:diffForSolve}
    \Biggl\{ & \sum_{\lambda, \varrho \in \underline k \times
      \underline l \atop \eta_{\lambda}+\mu_{\varrho} = z}
    \tilde{p}_{\eta_{\lambda}}(\uTheta)\tilde{q}_{\mu_{\varrho}}(\theta_1
    +(\eta_{\lambda})_1, \ldots, \theta_n +
    (\eta_{\lambda})_n)\gamma_{\eta_{\lambda}, \mu_{\varrho}}=\tilde h_z
    \notag\\[-8pt]
    &\quad |\quad z \in \ZZ^n,
    z^{(1)} \geq z \geq z^{(m)} \Biggr \}.
  \end{align}
  Moreover, a factorization of $h$ in $A_1$ corresponds to $\tilde
  q_{\mu_i}$ and $\tilde p_{\eta_j}$ for $(i,j) \in
  \underline{k}\times \underline{l}$ being polynomial solutions with
  bounds as stated in Lemma~\ref{lem:degBoundTheta}.
\end{theorem}
\begin{proof}
  We only sketch this technical proof.  Inspecting the product in
  \eqref{eq:diffForSolve}, we split it into its graded
  summands.  By repeated application of Lemma \ref{lem:rewriteKTheta},
  we arrive at the described set of equations via coefficient
  comparison. The degree bound has been established in Lemma~\ref{lem:degBoundTheta} above.
\end{proof}
\vspace*{-5pt}
\begin{corollary}
  The problem of factorizing a polynomial in the $n$th Weyl algebra
  can be solved via finding polynomial univariate solutions of degree
  at most $2\cdot \sum_{i =0}^{n} |\deg(h)_i|$ for a system of
  difference equations with polynomial coefficients, involving linear
  and quadratically nonlinear inhomogeneous equations.
\end{corollary}
\vspace*{-3pt}

As this part of the method is rather technical, let us illustrate it via
an example.

\vspace*{-5pt}
\begin{example}
  \label{ex:runningEx1}
  Let
  \begin{eqnarray}
    p &:=& \underbrace{\theta_1 \partial_2}_{=p_{[0,1]}} +
    \underbrace{(\theta_1 + 3)\theta_2}_{=p_{[0,0]}} +
    \underbrace{x_2}_{=p_{[0,-1]}},\nonumber\\
    q &:=&\underbrace{(\theta_1+4)x_1\partial_2}_{=q_{[-1,1]}} +
    \underbrace{x_1}_{=q_{[-1,0]}}+\underbrace{(\theta_1+1)x_1x_2}_{=q_{[-1,-1]}}
    \in A_2 \text{
      and} \nonumber
    \end{eqnarray}
    \begin{eqnarray}
    h &:=& pq=\theta_1(\theta_1+4)x_1\partial_2^2 \label{eq:hPart1} \\
    &+& (\theta_1(\theta_1-1)\theta_2 + 8\theta_1\theta_2 + \theta_1 +
    12\theta_2)x_1\partial_2  \label{eq:hPart2}\\
    &+& (\theta_1(\theta_1-1)\theta_2 + \theta_1^2 -\theta_1 +
    4\theta_1\theta_2 + 2\theta_1 + 7\theta_2)x_1\label{eq:hPart3}\\
    &+& (\theta_1(\theta_1-1)\theta_2 + 5\theta_1\theta_2 + 3\theta_2 +
    1)x_1x_2 \label{eq:hPart4}\\
    &+& (\theta_1+1)x_1x_2^2 \label{eq:hPart5}.
  \end{eqnarray}

  We have written every coefficient in terms of the $\theta_i$ already for better
  readability.

  By assumption, the only information we have about $p$ and $q$
  are the values of $p_{[0,1]} =:p_{\eta_1}$,
  $p_{[0,-1]}=:p_{\eta_3}$, $q_{[-1,1]}=:q_{\mu_1}$ and $q_{[-1,-1]} =:q_{\mu_l}$.
  Thus we have, using the above notation,
    $\tilde p_{\eta_1} = \theta_1$,  $\tilde p_{\eta_k} = 1$,
    $\tilde q_{\mu_1} = (\theta_1+4)$ and $\tilde q_{\mu_l} =
  (\theta_1+1)$.
  We set $k := l:= 3$, and it remains to solve for $\tilde{q}_{[-1,0]}$ and $\tilde{p}_{[0,0]}$.

  In $h$, every variable appears in degree 2, except from $x_1$, which
  appears in degree 3. That means that the degree bounds for
  $\theta_1$ and $\theta_2$ in $\tilde q_{\mu_i}$ can be set to be
  two. 

  The product of $(p_{\eta_1} + p_{\eta_2} + p_{\eta_3})(q_{\mu_1} +
  q_{\mu_2} + q_{\mu_3})$ with known values inserted is
  \begin{eqnarray}
    &&  \theta_1(\theta_1 + 4)x_1\partial_2^2 \label{eq:pqPart1}\\
    &+& (\theta_1 \tilde q_{\mu_2}  (\theta_1,\theta_2+1) + \tilde
    p_{\eta_2} (\theta_1+4))x_1\partial_2\label{eq:pqPart2}\\
    &+& (\theta_1(\theta_1+1)(\theta_2 +1) + (\theta_1 + 4)\theta_2 +
    \tilde p_{\eta_2} \tilde q_{\mu_2} )x_1\label{eq:pqPart3}\\
    &+& (\tilde q_{\mu_2} (\theta_1,\theta_2-1) + \tilde
    p_{\eta_2}(\theta_1+1))x_1 x_2\label{eq:pqPart4}\\
    &+& (\theta_1+1)x_1x_2^2\label{eq:pqPart5}.
  \end{eqnarray}

  The coefficients in $\KK[\uTheta]$ in the terms
  (\ref{eq:pqPart1})-(\ref{eq:pqPart5}) have to coincide with the
  respective coefficients in the terms
  (\ref{eq:hPart1})-(\ref{eq:hPart5}) for the factorization to be
  correct. The equations with respect to those coefficents are exactly
  the ones given in (\ref{eq:diffForSolve}).
\end{example}

\subsection{Determining the Rest of the Graded Parts}

There are many ways of dealing with finding solutions for the
system as described by the set (\ref{eq:diffForSolve}). The first way
would be to solve the appearing partial difference equations and
derive polynomial solutions. To the best of our knowledge, there is no
general algorithm for finding polynomial solutions of a system of nonlinear
difference equations (\cite{Abramov:1989}, \cite{Abramov:2012} and
\cite{Abramov:1995}). However, by Theorem \ref{MainThm} and Lemma
\ref{lem:degBoundTheta}, we are looking for bounded solutions, where
explicit bounds are given. This problem is clearly algorithmically
solvable.


Here, we present one of the possible approaches to solve the
commutative system of equations, which we also chose for the implementation.
 We give an
outline of the basic ideas here. A detailed description and discussion
will become subject of a journal version of this paper.

We begin by studying the equations as given in Theorem~\ref{MainThm}.

\vspace*{-5pt}
\begin{lemma}
  \label{lem:equationStudy}
  Let us sort the equations as given in the set stated in
  (\ref{eq:diffForSolve}) by the degree of the graded part they
  represent, from highest to lowest. Let moreover $\nu \in \NN$ be the
  number of those equations, and $\kappa$ be the number of all unknowns.
  We define $\chi_i$ for $i \in
  \underline{\nu}$ to be the number of $\tilde p_{\eta_\kappa}$ and $\tilde
  q_{\mu_\iota}$, $(\kappa, \iota) \in \underline{l}\times
  \underline{k}$,  appearing in equations
  $1,\ldots, i$. Then we have, for $i \leq \lceil \kappa/2 \rceil$,
  $\chi_i = 2\cdot (i-1)$.
  The same holds if we sort the equations from lowest to highest.
\end{lemma}
\begin{proof}
  The proof of this statement can be obtained using induction on
  $i$. We outline the main idea here. For $i=1$, we have the known
  equation $\tilde h_{z^{(1)}} = p_{\eta_1}q_{\mu_1} = \tilde
  p_{\eta_1}\tilde q_{\mu_1} (\theta_1 + (\eta_1)_1, \ldots,
  \theta_n + (\eta_1)_n)\gamma_{\eta_1, \mu_1}$, i.e. $\chi_1 =
  0$. For the next equation, as we regard the directly next lower
  homogeneous summand, only the directly next lower unknowns $\tilde
  p_{\eta_2}$ and $\tilde q_{\mu_2}$ appear, multiplied by
  $\tilde q_{\mu_1}$ resp. $\tilde p_{\eta_1}$. Hence, we get $\chi_2
  = 2$. This process can be iterated until $\chi_{\lceil \kappa/2
    \rceil} = \kappa$. An analogous argument can be used when the
  equations are sorted from lowest to highest.
\end{proof}

Using Lemma \ref{lem:equationStudy}, we can reduce the unknowns we
need to solve for to the $\tilde q_{\mu_i}$. Sorting the equations in
the set (\ref{eq:diffForSolve}) from highest to lowest, we can
rearrange them by putting the $\tilde p_{\eta_i}$ on the left hand
side and backwards substituting the appearing $\tilde p_{\eta_j}$ on
the respective right hand side by the formulae in the former
equations. The same can be done when sorting the equations from lowest
to highest, which lead to a second -- different -- set of
equations for the $\tilde p_{\eta_i}$. The remaining step is then to
concatenate the two respective descriptions for the $\tilde
p_{\eta_i}$ and then solve the resulting nonlinear system of equations
in the coefficients of the $\tilde q_{\mu_j}$ using e.g. Gr\"obner
bases \cite{Buchberger:1997}. We depict this process in the next
example.

\vspace*{-5pt}
\begin{example}
  \label{ex:finalInhomogEx}
  Let us consider $h = pq$ from Example \ref{ex:runningEx1}, using all
  notations that were introduced there.
  
   We assume that the given form of $\tilde p_{\eta_2}$ is
    $\tilde p_{\eta_2}  =  \tilde p_{\eta_2}^{(0)} + \tilde p_{\eta_2}^{(1)}\theta_1 +
    \tilde p_{\eta_2}^{(2)}\theta_1^2 + \tilde p_{\eta_2}^{(3)}\theta_2 +
    \tilde p_{\eta_2}^{(4)}\theta_1\theta_2
    + \tilde p_{\eta_2}^{(5)}\theta_1^2\theta_2 + \tilde p_{\eta_2}^{(6)}\theta_2^2 +
    \tilde p_{\eta_2}^{(7)}\theta_1\theta_2^2 +
    \tilde p_{\eta_2}^{(8)}\theta_1^2\theta_2^2,$
  and that $\tilde q_{\mu_2}$ has an analogous shape with coefficients
  $q_{\mu_2}^{(i)}$, where $\tilde p_{\eta_2}^{(i)}, \tilde q_{\mu_2}^{(i)} \in \KK$ for
  $i \in \underline{8}\cup \{0\}$.

  We use our knowledge of the form of $h$ and the product of $pq$ with
  unknowns as
  depicted (\ref{eq:pqPart1})-(\ref{eq:pqPart5}). Therefore, starting from the top and starting from the
  bottom, we obtain two expressions of $\tilde p_{\eta_2}$, namely
  \begin{eqnarray*}
  \tilde p_{\eta_2} &=& \frac{\theta_1(\theta_1-1)\theta_2 + 8\theta_1\theta_2 + \theta_1 +
  12\theta_2 - \theta_1 \tilde q_{\mu_2} 
  (\theta_1,\theta_2+1)}{\theta_1+4}\\
&=& \frac{\theta_1(\theta_1-1)\theta_2 + 5\theta_1\theta_2 + 3\theta_2 +
  1 - \tilde q_{\mu_2} (\theta_1,\theta_2-1)}{\theta_1+1}.
\end{eqnarray*}
  Thus, $\tilde q_{\mu_2}$ has to fulfil the equation
  \begin{eqnarray*}
    (\theta_1(\theta_1-1)\theta_2 + 8\theta_1\theta_2 + \theta_1 +
  12\theta_2 - \theta_1 \tilde q_{\mu_2} 
  (\theta_1,\theta_2+1))(\theta_1+1)\\
  = (\theta_1(\theta_1-1)\theta_2 + 5\theta_1\theta_2 + 3\theta_2 +
  1 - \tilde q_{\mu_2} (\theta_1,\theta_2-1))(\theta_1+4).
  \end{eqnarray*}

  Note here, that we could consider more equations that $\tilde
  q_{\mu_2}$ must fulfill, but we refrained from it in this example
  for the sake of brevity.

  Using coefficient comparison, one can form from this equation a
  nonlinear system of equations with the $\tilde q_{\mu_2}^{(i)}$,
  $i \in \underline{8}\cup \{0\}$, as indeterminates.
  The reduced Gr\"obner basis of this system is
  $\{\tilde q_{\mu_2}^{(0)}-1, \tilde q_{\mu_2}^{(1)}, \tilde q_{\mu_2}^{(2)},
  \ldots, \tilde q_{\mu_2}^{(8)} \},$
  which tells us, that $\tilde q_{\mu_2} =1$ and hence, 
  $\tilde p_{\eta_2} = (\theta_1+3)\theta_2$.
  Thus, we have exactly recovered both $p$ and $q$ in the factorization of $h$.
 The concrete original system is stated in Appendix \ref{app:polySystem}.
\end{example}

This approach of course raises the question, if those systems of
equations that we construct are over- resp. underdetermined. In the
latter case, we might end up with some ambiguity regarding the
solutions of the systems. The next lemma will show that our
construction in fact leads to an overdetermined system.
\vspace*{-5pt}
\begin{lemma} Let $\nu$ denote amount of the vectors in $\NN_0^n$,
that are in each component $t$ smaller or equal to \linebreak
$\min\{\deg_{x_t}(h),\deg_{\partial_t}(h) \}$. After the
reduction of the unknowns to the $\tilde q_{\mu_i}$ for $i \in
\{2,\ldots, l-1\}$, the amount of equations satisfied by the $\tilde
q_{\mu_i}$ will be between $2\cdot (l-1)\cdot \nu$ and $(l-1)^2\cdot
\nu$, and the amount of variables that we have to solve for is $(l-2)
\cdot \nu$.
\end{lemma}
\vspace*{-5pt}
\begin{proof} The number $(l-2)\cdot \nu$ is obvious for the number of
unknowns, as we have for every polynomial $\tilde q_{\mu_i}$ for $i
\in \{2,\ldots, l-1\}$ exactly $\nu$ unknown coefficients.

  In order to obtain expressions for our unknowns, we are considering
two times $l-1$ equations of the set in (\ref{eq:diffForSolve}),
namely $l-1$ equations starting from the bottom and $l-1$ equations
starting from the top. Note here, that we also consider the equation
for $\tilde q_{\mu_l}$ when starting from the top, and the equation
for $q_{\mu_1}$ when starting from the bottom, as we obtain more
equations fulfilled by the unknown variables in this way, where part
of it is known to us.  In the backwards substitution phase, we obtain
different products of the polynomials $\tilde q_{\mu_i}$. The amount
of terms in the $\theta_j$ for $j \in \undlnn$ of those products is
greater or equal to $2\cdot\nu$ and at most $(l-1)\cdot \nu$. This
leads to the claimed bounds.
\end{proof}
\vspace*{-5pt}
\subsection{Application to Weyl Algebras with\\[2pt] Rational Coefficients}

In practice, one is often interested in differential equations 
over the field of rational functions in the indeterminates $x_i$. 
We refer to the corresponding operator algebras as the \textbf{rational Weyl
  algebras}. We have the same commutation rules there, but with
extension to the case where $x_i$ appears in the denominator. These algebras 
can be recognized as Ore localization of polynomial Weyl algebras with respect
to the multiplicatively closed Ore set $S=\KK[x_1,\ldots,x_n]\setminus\{0\}$.

Unlike in the polynomial Weyl algebra, an infinite number of
nontrivial factorizations of an element is possible. The
easiest example is the polynomial $\partial_1^2 \in A_1$, 
having nontrivial factorizations \mbox{$(\partial_1+ \frac{1}{x_1+c})(\partial_1 - \frac{1}{x_1+c})$} for all $c \in \KK$; the only polynomial factorization is $\partial\cdot
\partial$.  Thus, at first glance, the factorization problem in both
the rational and the polynomial Weyl algebras seems to be distinct in
general. But there are still many things in common.

Consider the more general case of localization of Ore algebras. In
what follows, we denote by $S \subset R$ the \textbf{denominator set}
of an arbitrary localization of a Noetherian integral domain $R$. For properties
that $S$ has to fulfil and calculation rules of elements in $S^{-1}R$
please consider \cite{Bueso:2003}, Chapter 8.
Let us clarify the connection between factorizations in $S^{-1}R$ and factorizations in
$R$. 

\vspace*{-4pt}
\begin{theorem}
\label{thrm:liftingrationalfact}
Let $h$ be an element in $S^{-1}R\setminus\{0\}$. Suppose, that 
$h = h_1 \cdots h_m$, $m \in \NN$, $h_i \in S^{-1}R$ for $i \in \underline{m}$. Then there exists $q\in S$ and $\tilde h_1, \ldots, \tilde h_m \in R$, such
that $qh = \tilde h_1\cdots \tilde h_m.$
\end{theorem}

\vspace*{-3pt}
Thus, by clearing denominators in an irreducible element in $S^{-1}R$ one obtains an
irreducible element in $R$. The other direction does not hold in general. 
However, one can use our algorithms in a pre-processing step of finding factorization
over $S^{-1}R$. In particular, a reducible element of $R$ is necessarily reducible
over $S^{-1}R$. 

The theorem says that we can lift any factorization from the ring
$S^{-1}R$ to a factorization in $R$ by a left multiplication with an
element of $S$. This means that in our case, where
$S=\KK[x_1,\ldots,x_n]\setminus\{0\}$, it suffices to multiply a
polynomial $h$ by a suitable element in $\KK[x_1,\ldots, x_n]$ in
order to obtain a representative of a rational factorization. Finding
this element is subject of future research. As we already have shown
in \cite{heinle2013factorization}, a polynomial factorization of an
element in $A_n$ is often more readable than the factorization
produced by rational factorization methods. Thus a pre-computation
that finds such a premultiplier so that we can just perform polynomial
factorization would be a beneficial ansatz in the rational
factorization.

\begin{example}
  Consider the polynomial
  \(h := \partial_1^3 - x_1\partial_1 -2 \in A_1.\)
  $h$ is irreducible in $A_1$, but in the first rational Weyl algebra, we obtain a factorization given by
  \(
  (\partial_1 + \frac{1}{x_1})(\partial_1^2 -\frac{1}{x_1}\partial_1 -
  x_1).
  \)
  If we multiply $h$ by $x_1$ from the left, our factorization method
  reveals two different factorizations. The first one is $x_1\cdot h$
  itself, and the second one is given by
  \(\partial_1 \cdot (x_1\partial_1^2 - x_1^2 - \partial_1),\)
  which represents the rational factorization in the sense of Theorem \ref{thrm:liftingrationalfact}.
\end{example}

\subsection{Application to Shift Algebras (with Rational Coefficients)}

  With the help of the Lemma \ref{lem:rewriteKTheta} one can see that
  $\mathcal{S}_n$ is a subalgebra of the $n$th Weyl algebra $A_n$
via the following homomorphism of $\KK$-algebras:
  $\iota: \mathcal{S}_n \to A_n, \quad x_i \mapsto \theta_i, \
  s_j \mapsto \partial_j.$

  One can easily prove that $\iota$ is, in fact, a monomorphism. This
  observation leads to the following result, which tells us that we do
  not have to consider the algebra $\mathcal{S}_n$ separately when
  dealing with factorization of its elements.

\vspace*{-6pt}
\begin{corollary}
  \label{lem:reduceShiftToWeyl}
The factorization problem for a polynomial $p \in \mathcal{S}_n$ can be
 obtained from the solution of a factorization problem of $\iota(p) \in A_n$
  by refining.
\end{corollary}

\vspace*{-8pt}
Theorem
\ref{thrm:liftingrationalfact} also applies to the rational shift
algebra. Thus, the approach to lift factorizations in the shift
algebras with rational coefficients can also be applied here. The
remaining research is also here to find suitable elements in
$\KK[x_1,\ldots, x_n]$ for pre-multiplication.

\vspace*{-5pt}
\section{Implementation and Timings}
\label{sctn:implAndTime}

We have implemented the described method for $A_n$ in the computer
algebra system \textsc{Singular}. Our goal was to test the performance
of our approach and the versatility of the results in practice and
compare it to given implementations. Our implementation is in a
complete but experimental stage, and we see potential for optimization
in several areas.

The implementation extends the library \texttt{ncfactor.lib}, which
contains the functionality to factorize polynomials in the first Weyl algebra, the first shift
algebra and graded polynomials in the first $q$-Weyl algebra. The actual library
is distributed with \textsc{Singular} since version 3-1-3.

In the following examples, we consider different polynomials and
present the resulting factorizations and timings. Our function to
factorize polynomials in the $n$th Weyl algebra is written to solve
problem (ii) as given in the introduction, i.e. finding
all possible factorizations of a given polynomial. All computations
were done using \textsc{Singular} version 3-1-6. We compare our
performance and our outputs to \textsc{REDUCE} version 3.8. There, we
use the function \texttt{nc\_factorize\_all} in the library
\texttt{NCPOLY}. The calculations were run on a on a computer
with a 4-core Intel CPU (Intel\textregistered  Core\texttrademark
i7-3520M CPU with 2.90GHz, 2 physical cores, 2 hardware threads, 32K
L1[i,d], 256K L2, 4MB L3 cache) and 16GB RAM.

In order to make the tests reproducible, we used the \textsc{SDEval}
\cite{heinle2013symbolicdata} framework, created for the
\textsc{Symbolic Data} project \cite{Bachmann:2000}, for our
benchmarking. The functions of \textsc{Symbolic Data} as well as the
data are free to use. In such a way our comparison is easily reproducible by any other person.


  Our set of examples is given by
  \begin{eqnarray*}
    h_1 &:=& (\partial_1+1)^2(\partial_1 + x_1\partial_2) \in A_2,\\
    h_2 &:=& (\theta_1 \partial_2 + (\theta_1 + 3)\theta_2 +
    x_2)\cdot\\ 
    &&\hspace*{8pt}((\theta_1+4)x_1\partial_2 + x_1+(\theta_1+1)x_1x_2)\in
    A_2,
  \end{eqnarray*}
  \begin{eqnarray*}
   h_3 & := & x_1x_2^2x_3^3\partial_1\partial_2^2+x_2x_3^3\partial_2 \in
  A_3,\\
   h_4 &:=& (x_1^2\partial_1+x_1x_2\partial_2)(\partial_1\partial_2 +
  \partial_1^2\partial_2^2x_1x_2) \in A_2.
  \end{eqnarray*}
  The polynomial $h_1$ can be found in \cite{landau1902satz}, the
  polynomial $h_2$ is
  the polynomial from Example \ref{ex:finalInhomogEx} and the last two
  polynomials are graded polynomials.

  Our implementation in \textsc{Singular} managed to factor all the
  polynomials that are listed above. For $h_1$, it took 2.83s to find
  two distinct factorizations. Besides
  the given one above, we have
  $h_1 =
  (x_1\partial_1\partial_2+\partial_1^2+x_1\partial_2+\partial_1+2\partial_2)(\partial_1+1).$
  In order to factorize $h_2$, \textsc{Singular} took 23.48s to find
  three factorizations.
  For the graded polynomials $h_3$ and $h_4$, our implementation
  finished its computations as expected quickly (0.46s and 0.32s) and
  returned 60 distinct factorizations for each $h_3$ and $h_4$.

  \textsc{REDUCE} only terminated for $h_1$ (within two hours).  For $h_1$
  it returned 3 different factorizations (within 0.1s), and one of the factorizations
  contained a reducible factor. For $h_2,h_3$ and $h_4$, we cancelled
  the process after two hours.

Factoring $\ZZ$-graded polynomials in the first Weyl algebra was already
timed and compared with several implementations on various examples in
\cite{heinle2013factorization}. The comparison there also included the
functionality in the computer algebra system \textsc{Maple} for
factoring polynomials in the first Weyl algebra with rational coefficients.

The next example shows the performance of our implementation for the
first Weyl algebra.

\vspace*{-5pt}
\begin{example}
This example is taken from \cite{Koepf:1998}, page 200. We consider
\(
h := (x_1^4-1)x_1\partial_1^2 + (1+7x_1^4)\partial_1 + 8x_1^3.
\)
Our implementation takes 0.75 seconds to find 12 distinct
factorizations in the algebra $A_1$.  \textsc{Maple} 17, using
\texttt{DFactor} from the \texttt{DETools} package, takes the same
amount of time and reveals one factorization in the first Weyl algebra
with rational coefficients.  \textsc{REDUCE} outputs 60 factorizations
in $A_1$ after 3.27s.  However, these factorizations contain factorizations
with reducible factors.  After factoring such cases and removing
duplicates from the list, the number of different factorizations
reduced to 12.
\end{example}

\vspace*{-12pt}
\section{Conclusions}
\label{sctn:conclusion}
An approach to factoring polynomials in the operator algebras $A_n$,
$Q_n$ and $\mathcal{S}_n$ based on nontrivial $\ZZ^n$-gradings has
been presented, and an experimental implementation has been
evaluated. We have shown that the set of polynomials that we can
factorize using our technique in a feasible amount of time has been greatly
extended. Especially for $\ZZ^n$-graded polynomials, we have shown that the
problem of finding all nontrivial factorizations in $A_n$ resp. $Q_n$ can be reduced
to commutative factorization in multivariate rings and some basic combinatorics. 
Thus, the performance of the factorization algorithm
regarding graded polynomials is dominated by the performance of the
commutative factorization algorithm that is available.

Our future work consists of implementing the remaining functionalities
into \texttt{ncfactor.lib}. Furthermore, it would be interesting to
extend our technique to deal with the factorization problem in $A_n$
to polynomials in $Q_n$. Additionally, there exist many other
operator algebras, and it would be interesting to investigate to what
extent we can use the described methodology there.

Applying our techniques for the factorization problem in the case of
algebras with coefficients in rational functions is also interesting, albeit
more involved. Amongst other problems, in that case infinitely many different
factorizations can occur. One has to find representatives of
parametrized factorizations, and use these to obtain a factorization
in the polynomial sense. This approach could be beneficial, and it has
been developed in \cite{heinle2013factorization}.

\vspace*{-7pt}
\section{Acknowledgments}
We would like to thank to Dima Grigoriev for discussions on the
subject, and to Mark van Hoeij for his expert opinion.  We are
grateful to Wolfram Koepf and Martin Lee for providing us with
interesting examples
and to Michael Singer, Shaoshi Chen and Daniel Rettstadt
for sharing with us interesting points of view on our research problems.

We would like to express our gratitude to the German Academic Exchange 
Service
DAAD for funding our project in the context of the German-Canadian PPP program. 
\vspace*{2pt}
{\small
\bibliographystyle{abbrv}
\bibliography{facnthweyl}}

\vfill
\appendix

\section{Commutative Polynomial System of Equations}
\label{app:polySystem}
The commutative polynomial system of equations that is formed in
Example \ref{ex:finalInhomogEx} is given as follows.

\begin{eqnarray*}
&&\{-\tilde q_{\mu_2}^{(8)}, -\tilde q_{\mu_2}^{(7)}-2\tilde q_{\mu_2}^{(8)},
-\tilde q_{\mu_2}^{(6)}-\tilde q_{\mu_2}^{(7)}-\tilde q_{\mu_2}^{(8)},
-\tilde q_{\mu_2}^{(5)},
-\tilde q_{\mu_2}^{(4)}-2\tilde q_{\mu_2}^{(5)}\\
&&-4\tilde q_{\mu_2}^{(8)},
-\tilde q_{\mu_2}^{(3)}-\tilde q_{\mu_2}^{(4)}-\tilde q_{\mu_2}^{(5)}-2\tilde q_{\mu_2}^{(7)},
-\tilde q_{\mu_2}^{(2)}+4\tilde q_{\mu_2}^{(8)},
-\tilde q_{\mu_2}^{(1)}-2\tilde q_{\mu_2}^{(2)}\\
&&-4\tilde q_{\mu_2}^{(5)}+4\tilde q_{\mu_2}^{(7)}-8\tilde q_{\mu_2}^{(8)},
-\tilde q_{\mu_2}^{(0)}-\tilde q_{\mu_2}^{(1)}-\tilde
q_{\mu_2}^{(2)}-2\tilde q_{\mu_2}^{(4)}+4\tilde q_{\mu_2}^{(6)}\\
&&-4\tilde q_{\mu_2}^{(7)}+4\tilde q_{\mu_2}^{(8)}+1,
-4\tilde q_{\mu_2}^{(2)}+4\tilde q_{\mu_2}^{(4)}-8\tilde q_{\mu_2}^{(5)},
-2\tilde q_{\mu_2}^{(1)}+4\tilde q_{\mu_2}^{(3)}\\
&&-4\tilde q_{\mu_2}^{(4)}+4\tilde q_{\mu_2}^{(5)},
4\tilde q_{\mu_2}^{(2)},
4\tilde q_{\mu_2}^{(1)}-8\tilde q_{\mu_2}^{(2)},
4\tilde q_{\mu_2}^{(0)}-4\tilde q_{\mu_2}^{(1)}+4\tilde
q_{\mu_2}^{(2)}-4,\\
&&
(\tilde q_{\mu_2}^{(8)})^2,
2\tilde q_{\mu_2}^{(7)}\tilde q_{\mu_2}^{(8)}+2(\tilde q_{\mu_2}^{(8)})^2,
2\tilde q_{\mu_2}^{(6)}\tilde q_{\mu_2}^{(8)}+(\tilde
q_{\mu_2}^{(7)})^2+3\tilde q_{\mu_2}^{(7)}\tilde q_{\mu_2}^{(8)}\\
&&
+(\tilde
q_{\mu_2}^{(8)})^2,
2\tilde q_{\mu_2}^{(6)}\tilde q_{\mu_2}^{(7)}+2\tilde
q_{\mu_2}^{(6)}\tilde q_{\mu_2}^{(8)}+(\tilde q_{\mu_2}^{(7)})^2+\tilde
q_{\mu_2}^{(7)}\tilde q_{\mu_2}^{(8)},
(\tilde q_{\mu_2}^{(6)})^2\\
&&
+\tilde q_{\mu_2}^{(6)}\tilde
q_{\mu_2}^{(7)}+\tilde q_{\mu_2}^{(6)}\tilde q_{\mu_2}^{(8)},
2\tilde q_{\mu_2}^{(5)}\tilde q_{\mu_2}^{(8)},
2\tilde q_{\mu_2}^{(4)}\tilde q_{\mu_2}^{(8)}+2\tilde q_{\mu_2}^{(5)}\tilde q_{\mu_2}^{(7)}+4\tilde q_{\mu_2}^{(5)}\tilde q_{\mu_2}^{(8)}\\
&&-\tilde q_{\mu_2}^{(8)},
2\tilde q_{\mu_2}^{(3)}\tilde q_{\mu_2}^{(8)}+2\tilde q_{\mu_2}^{(4)}\tilde q_{\mu_2}^{(7)}+3\tilde
q_{\mu_2}^{(4)}\tilde q_{\mu_2}^{(8)}+2\tilde q_{\mu_2}^{(5)}\tilde
q_{\mu_2}^{(6)}+3\tilde q_{\mu_2}^{(5)}\tilde q_{\mu_2}^{(7)}\\
&&+2\tilde
q_{\mu_2}^{(5)}\tilde q_{\mu_2}^{(8)}-\tilde q_{\mu_2}^{(7)},
2\tilde q_{\mu_2}^{(3)}\tilde q_{\mu_2}^{(7)}+2\tilde
q_{\mu_2}^{(3)}\tilde q_{\mu_2}^{(8)}+2\tilde q_{\mu_2}^{(4)}\tilde
q_{\mu_2}^{(6)}+2\tilde q_{\mu_2}^{(4)}\tilde q_{\mu_2}^{(7)}\\
&&+\tilde
q_{\mu_2}^{(4)}\tilde q_{\mu_2}^{(8)}+2\tilde q_{\mu_2}^{(5)}\tilde
q_{\mu_2}^{(6)}+\tilde q_{\mu_2}^{(5)}\tilde q_{\mu_2}^{(7)}-\tilde q_{\mu_2}^{(6)},
2\tilde q_{\mu_2}^{(3)}\tilde q_{\mu_2}^{(6)}+\tilde
q_{\mu_2}^{(3)}\tilde q_{\mu_2}^{(7)}\\
&& +\tilde q_{\mu_2}^{(3)}\tilde
q_{\mu_2}^{(8)}+\tilde q_{\mu_2}^{(4)}\tilde q_{\mu_2}^{(6)}+\tilde
q_{\mu_2}^{(5)}\tilde q_{\mu_2}^{(6)},
2\tilde q_{\mu_2}^{(2)}\tilde q_{\mu_2}^{(8)}+(\tilde q_{\mu_2}^{(5)})^2,
2\tilde q_{\mu_2}^{(1)}\tilde q_{\mu_2}^{(8)}\\
&&+2\tilde
q_{\mu_2}^{(2)}\tilde q_{\mu_2}^{(7)}+4\tilde q_{\mu_2}^{(2)}\tilde
q_{\mu_2}^{(8)}+2\tilde q_{\mu_2}^{(4)}\tilde q_{\mu_2}^{(5)}+2(\tilde
q_{\mu_2}^{(5)})^2-\tilde q_{\mu_2}^{(5)}-7\tilde q_{\mu_2}^{(8)},\\
&&
2\tilde q_{\mu_2}^{(0)}\tilde q_{\mu_2}^{(8)}+2\tilde
q_{\mu_2}^{(1)}\tilde q_{\mu_2}^{(7)}+3\tilde q_{\mu_2}^{(1)}\tilde
q_{\mu_2}^{(8)}+2\tilde q_{\mu_2}^{(2)}\tilde q_{\mu_2}^{(6)}+3\tilde q_{\mu_2}^{(2)}\tilde q_{\mu_2}^{(7)}\\
&&+2\tilde
q_{\mu_2}^{(2)}\tilde q_{\mu_2}^{(8)}+2\tilde q_{\mu_2}^{(3)}\tilde
q_{\mu_2}^{(5)}+(\tilde q_{\mu_2}^{(4)})^2+3\tilde q_{\mu_2}^{(4)}\tilde
q_{\mu_2}^{(5)}-\tilde q_{\mu_2}^{(4)}+(\tilde q_{\mu_2}^{(5)})^2\\
&&-7\tilde q_{\mu_2}^{(7)}-\tilde q_{\mu_2}^{(8)},
2\tilde q_{\mu_2}^{(0)}\tilde q_{\mu_2}^{(7)}+2\tilde
q_{\mu_2}^{(0)}\tilde q_{\mu_2}^{(8)}+2\tilde q_{\mu_2}^{(1)}\tilde
q_{\mu_2}^{(6)}+2\tilde q_{\mu_2}^{(1)}\tilde q_{\mu_2}^{(7)}\\
&&+\tilde q_{\mu_2}^{(1)}\tilde q_{\mu_2}^{(8)}+2\tilde
q_{\mu_2}^{(2)}\tilde q_{\mu_2}^{(6)}+\tilde q_{\mu_2}^{(2)}\tilde
q_{\mu_2}^{(7)}+2\tilde q_{\mu_2}^{(3)}\tilde q_{\mu_2}^{(4)}+2\tilde
q_{\mu_2}^{(3)}\tilde q_{\mu_2}^{(5)}-\tilde q_{\mu_2}^{(3)}\\&&
+(\tilde q_{\mu_2}^{(4)})^2+\tilde q_{\mu_2}^{(4)}\tilde q_{\mu_2}^{(5)}-7\tilde q_{\mu_2}^{(6)}-\tilde q_{\mu_2}^{(7)},
2\tilde q_{\mu_2}^{(0)}\tilde q_{\mu_2}^{(6)}+\tilde
q_{\mu_2}^{(0)}\tilde q_{\mu_2}^{(7)}+\tilde q_{\mu_2}^{(0)}\tilde
q_{\mu_2}^{(8)}\\&&
+\tilde q_{\mu_2}^{(1)}\tilde q_{\mu_2}^{(6)}+\tilde q_{\mu_2}^{(2)}\tilde q_{\mu_2}^{(6)}+(\tilde q_{\mu_2}^{(3)})^2+\tilde q_{\mu_2}^{(3)}\tilde q_{\mu_2}^{(4)}+\tilde q_{\mu_2}^{(3)}\tilde q_{\mu_2}^{(5)}-\tilde q_{\mu_2}^{(6)},
\\&&2\tilde q_{\mu_2}^{(2)}\tilde q_{\mu_2}^{(5)},
2\tilde q_{\mu_2}^{(1)}\tilde q_{\mu_2}^{(5)}+2\tilde q_{\mu_2}^{(2)}\tilde q_{\mu_2}^{(4)}+4\tilde q_{\mu_2}^{(2)}\tilde q_{\mu_2}^{(5)}-\tilde q_{\mu_2}^{(2)}-7\tilde q_{\mu_2}^{(5)}-12\tilde q_{\mu_2}^{(8)},\\
&&2\tilde q_{\mu_2}^{(0)}\tilde q_{\mu_2}^{(5)}+2\tilde q_{\mu_2}^{(1)}\tilde q_{\mu_2}^{(4)}+3\tilde
q_{\mu_2}^{(1)}\tilde q_{\mu_2}^{(5)}-\tilde q_{\mu_2}^{(1)}+2\tilde
q_{\mu_2}^{(2)}\tilde q_{\mu_2}^{(3)}+3\tilde q_{\mu_2}^{(2)}\tilde
q_{\mu_2}^{(4)}\\
&&+2\tilde q_{\mu_2}^{(2)}\tilde q_{\mu_2}^{(5)}-7\tilde q_{\mu_2}^{(4)}-\tilde q_{\mu_2}^{(5)}-12\tilde q_{\mu_2}^{(7)},
2\tilde q_{\mu_2}^{(0)}\tilde q_{\mu_2}^{(4)}+2\tilde
q_{\mu_2}^{(0)}\tilde q_{\mu_2}^{(5)}-\tilde q_{\mu_2}^{(0)}\\
&&+2\tilde
q_{\mu_2}^{(1)}\tilde q_{\mu_2}^{(3)}+2\tilde q_{\mu_2}^{(1)}\tilde q_{\mu_2}^{(4)}+\tilde
q_{\mu_2}^{(1)}\tilde q_{\mu_2}^{(5)}+2\tilde q_{\mu_2}^{(2)}\tilde
q_{\mu_2}^{(3)}+\tilde q_{\mu_2}^{(2)}\tilde q_{\mu_2}^{(4)}-7\tilde
q_{\mu_2}^{(3)}\\&& -\tilde q_{\mu_2}^{(4)}-12\tilde q_{\mu_2}^{(6)}+1,
2\tilde q_{\mu_2}^{(0)}\tilde q_{\mu_2}^{(3)}+\tilde q_{\mu_2}^{(0)}\tilde q_{\mu_2}^{(4)}+\tilde q_{\mu_2}^{(0)}\tilde q_{\mu_2}^{(5)}+\tilde q_{\mu_2}^{(1)}\tilde q_{\mu_2}^{(3)}\\
&&+\tilde q_{\mu_2}^{(2)}\tilde q_{\mu_2}^{(3)}-\tilde q_{\mu_2}^{(3)},
(\tilde q_{\mu_2}^{(2)})^2,
2\tilde q_{\mu_2}^{(1)}\tilde q_{\mu_2}^{(2)}+2(\tilde q_{\mu_2}^{(2)})^2-7\tilde q_{\mu_2}^{(2)}-12\tilde q_{\mu_2}^{(5)},\\
&&
2\tilde q_{\mu_2}^{(0)}\tilde q_{\mu_2}^{(2)}+(\tilde
q_{\mu_2}^{(1)})^2+3\tilde q_{\mu_2}^{(1)}\tilde
q_{\mu_2}^{(2)}-7\tilde q_{\mu_2}^{(1)}+(\tilde q_{\mu_2}^{(2)})^2-\tilde q_{\mu_2}^{(2)}-12\tilde q_{\mu_2}^{(4)},\\
&& 
2\tilde q_{\mu_2}^{(0)}\tilde q_{\mu_2}^{(1)}+2\tilde
q_{\mu_2}^{(0)}\tilde q_{\mu_2}^{(2)}-7\tilde q_{\mu_2}^{(0)}+(\tilde
q_{\mu_2}^{(1)})^2+\tilde q_{\mu_2}^{(1)}\tilde q_{\mu_2}^{(2)}-\tilde q_{\mu_2}^{(1)}-12\tilde q_{\mu_2}^{(3)}\\
&&+7,
(\tilde q_{\mu_2}^{(0)})^2+\tilde q_{\mu_2}^{(0)}\tilde q_{\mu_2}^{(1)}+\tilde q_{\mu_2}^{(0)}\tilde q_{\mu_2}^{(2)}-\tilde q_{\mu_2}^{(0)},
-12\tilde q_{\mu_2}^{(1)},
-12\tilde q_{\mu_2}^{(0)}+12\}.
\end{eqnarray*}

\end{document}